\documentclass[11pt]{amsart}

 \usepackage[margin=1in]{geometry}

\newcommand{\private}[1]{}
\usepackage{amssymb, amsmath, amscd, amsthm, color, epsfig, url, tikz, graphicx
}

\usepackage[all]{xy}          
\xyoption{dvips}              

\makeatletter
\renewcommand\l@subsection{\@tocline{2}{0pt}{2pc}{5pc}{}}
\makeatother

\newcommand{\R}{{\mathbb R}}
\newcommand{\Z}{{\mathbb Z}}

\theoremstyle{plain}
\newtheorem{thm}{Theorem}[section]
\newtheorem{prop}[thm]{Proposition}

\newtheorem{cor}[thm]{Corollary}

\theoremstyle{definition}
\newtheorem{defin}[thm]{Definition}
\newtheorem{example}[thm]{Example}
\newtheorem{def/ex}[thm]{Definition/Example}

\theoremstyle{remark}
\newtheorem{rem}[thm]{Remark}

\newcommand{\refT}[1]{Theorem~\ref{T:#1}}
\newcommand{\refC}[1]{Corollary~\ref{C:#1}}
\newcommand{\refP}[1]{Proposition~\ref{P:#1}}
\newcommand{\refD}[1]{Definition~\ref{D:#1}}

\makeatletter
\@namedef{subjclassname@2020}{%
  \textup{2020} Mathematics Subject Classification}
\makeatother

\begin{document}
\title{Weighted simple games and the topology of simplicial complexes}

\author{Anastasia Brooks}
\address{Department of Mathematics, Wellesley College, Wellesley, MA}
\email{ab15@wellesley.edu}

\author{Franjo \v Sar\v cevi\'c}
\address{Faculty of Civil Engineering, University of Rijeka, Croatia}
\email{franjo.sarcevic@uniri.hr}

\author{Ismar Voli\'c}
\address{Department of Mathematics, Wellesley College, Wellesley, MA}
\email{ivolic@wellesley.edu}
\urladdr{ivolic.wellesley.edu}

\subjclass[2020]{91F10, 91B12, 05E45}
\keywords{simple games, weighted voting games, feasible coalition, quantification of power, power index, Banzhaf index, Shapley-Shubik index, simplicial complex, U.N.~Security Council, U.S.~Electoral College, Bosnia-Herzegovina Parliament}

\thanks{\textbf{Acknowledgements.} The second author is grateful to the University of Seville and the University of Sarajevo where parts of this paper were written. The third author would like to thank the Simons Foundation for its support. He is also grateful to University of Sarajevo where parts of this paper were written. The authors also thank Ava Mock who worked on this project in its early stages.}

\begin{abstract}
We use simplicial complexes to model simple games as well as weighted voting games where certain coalitions are considered impossible. Topological characterizations of various ideas from simple games are provided, as are the expressions for Banzhaf and Shapley-Shubik power indices for weighted games. We calculate the indices in several examples of weighted voting games with unfeasible coalitions, including the U.S.~Electoral College and the Parliament of Bosnia-Herzegovina.
\end{abstract}

\maketitle


\tableofcontents


\parskip=5pt
\parindent=0cm




\section{Introduction}\label{S:Intro}
Simplicial complexes are a natural tool for modeling structures in which there exist interactions  between objects. In recent years, the use of simplicial complexes to model such interactions has become increasingly common in various fields, including in social sciences (and specifically in social choice theory) where simplicial complexes have been use to model social communication and opinion dynamics \cite{HG21, IPBL:SimplicialContagion, WZLS:SocialSimplicial}.

Modeling political structures with simplicial complexes, where the voters or agents are represented by vertices and coalitions by simplices, has enabled the import of a number of topological concepts into the picture \cite{AK:Conflict, MV:Politics}. In this paper, we will take this setup further by applying the simplicial complex framework to simple games and the calculation of power indices for weighted voting games.

Simple games correspond to simplicial complexes via the observation that the losing coalitions form a complex. This is because of the requirement that any subset of a losing coalition is losing, which is precisely the defining property of a simplicial complex. This correspondence does not seem to have been explored in the literature, and we hope our work initiates further investigation. A notable result in this direction is \refT{CompleteWeightedGames} (and its consequence \refC{CompleteHomologyWeighted}) which provides a topological characterization of a certain type of a weighted voting game.

One of the reasons weighted voting games are an important class of simple games is that the theory of quantification of power, or \emph{power indices}, can be applied to them. Power indices are the second way in which we introduce topology into simple game theory. In general, these indices measure a voter's ability to change the outcome of a game. Two of the most classical indices are the Banzhaf index \cite{Ban65} and the Shapley-Shubik index \cite{ShSh54}.\footnote{There are other power indices, such as Deegan-Packel and Public Good, but Banzhaf and Shapley-Shubik are the most familiar.} However, in their most basic definition, neither of them can account for certain nuances of real-life voting schemes. Both assume an equal likelihood of all coalitions being formed, but reality suggests that some coalitions might be impossible. 

The literature that builds such considerations into modifications of power indices is extensive, but none of it uses our approach via simplicial complexes. Incorporating \emph{quarreling} coalitions into theory goes back to the 1970s \cite{K:Quarreling}, with more recent work including that of Schmidtchen and Steunenberg \cite{SS:PreferencePower} who take institutional considerations into account and K\'oczy \cite{K:QuarrelingPower} who incorporates voters' preferences and strategies (the starting point being the ``paradox'' that refusing to enter a coalition can increase a voter's power). Some of this theory is recounted in Section 7.5. of the classic book \cite{FM:VotingPower}.

The situation when the communication between voters is limited goes back to Myerson \cite{M:GraphsCooperation} who uses graphs to encode which voters are in communication with one another. There is also a complementary approach where voter incompatibilites are kept track of by a graph. Now an edge between two players means that they can never be in a coalition together. Such games are sometimes called \emph{$I$-restricted} \cite{carreras:1991, bergatinos:1993, yakuba:2008}. Power indices for systems with such restrictions have also been studied \cite{am:2015} and use generating functions to compute them. 

One point of view on our work is that we expand on the notion that a graph captures voter (in)compatibilities. A graph is an example of a simplicial complex, but it is limited in that it only knows about relationships between pairs of voters via its edges. If, for example, three voters are compatible, a graph cannot capture this information. A simplicial complex, on the other hand, can, and this three-fold compatibility would be represented by a triangle. Our formulas for the Banzhaf and Shapley-Shubik indices are thus in a way generalizations of those for graph restricted weighted voting games.

Further situation where there is a gradation of cooperation between voters is done by Owen \cite{O86}, and Borm, Owen, and Tijs \cite{MOT:PositionValue}. 
Owen \cite{O77, O81} modifies the power indices by introducing non-equal probability assumptions for the coalitions. A further advance in that direction was made by Edelman \cite{E97}, whose work was extended by Perlinger and Berg \cite{P00, BP00} and Mazurkiewicz and Mercik \cite{MM05}. Recent literature on these topics includes \cite{MS18, deM20, Kur20}. In particular, the authors of \cite{MS18} give a modification of the Banzhaf, Shapley-Shubik, and other power indices to take into account the profiles of the political parties as the main factor for forming a winning coalition. Aleskerov \cite{A:PowerPreferences} extends the Banzhaf index to incorporate the intensity of the voters' desire to form coalitions. Freixas \cite{F:BanzhafMultichoice}, building on work by \cite{FZ:WeightedMultichoice}, extends the Banzhaf index to situations where voters can choose from a finite set of  ordered actions or approvals.

%

The simplicial complex model for cooperative games which tries to capture a more nuanced coalition dynamic is fairly novel. The most relevant work is that of Martino \cite{Martino:GamesComplexes} who sets up the dictionary between coalitions and simplices. Unfeasible coalitions are captured by the absence of simplices that would have been spanned by the vertices representing those voters. This work generalizes the study of matroid models of simple games \cite{B:CooperativeCombinatorics, BDJ:MatroidStatic, BDJ:MatroidDynamic} as matroids are special cases of simplicial complexes. Another appearance of simplicial complexes in simple games can be found in \cite{GP:GameManifold}, where the authors construct a combinatorial manifold associated to the simplicial complex modeling a simple game. Elsewhere in cooperative game theory, simplicial complexes appear in the study of combinatorial \cite{FHN:GameComplex1, FHN:GameComplex2, FHN:GameComplex} and snowdrift games \cite{XFZX:SnowdriftGame}.

None of this work, however, uses the topology of simplicial complexes as a way to model simple games or reinterpret power indices in the presence of unfeasible coalitions. The goal of this paper is to do so, and in Section \ref{S:SimplicialSimple} we study the topology of simple games while in Section \ref{S:Computing} we supply formulas for the Banzhaf and Shapley-Shubik indices as a count of certain types of simplices. 

A few examples of weighted games (UN Security Council, U.S.~Electoral College, Bosnian-Herzegovinian Parliament) are also provided. The calculation of the Banzhaf index for the U.S.~Electoral College, considering the blue and red wall states as single voters, is particularly revealing. The calculations were performed using a Python code we wrote.\footnote{The code is available at \url{https://tinyurl.com/y2ev5pyz}.} The code does not improve on existing computational tools since it checks over all possible coalitions first and then compares those against the unfeasible ones, but it contains some potentially useful methods for filtering coalitions. For the discussion of computational complexity of calculating power indices, see \cite{MM:AlgorithmsIndices}.

Even though our formulas for the Banzhaf and Shapley-Shubik indices are new, they are largely just straightforward consequences of the translation from the language of simple games to that of simplicial complexes. Thus the results of this paper can be regarded as the beginning of the exploration of the correspondence between these two seemingly unrelated fields. In this spirit, Section \ref{S:FurtherResearch} offers various potential directions for further work, much of it based on recasting some standard constructions from the topology and geometry of simplicial complexes in the language of simple games. Among other things, we speculate that operations such as the wedge, cone, and suspension, as well as simplicial maps between complexes and the homology of simplicial complexes can provide insight and be interpreted in game-theoretic terms.  The framework of category theory, imported from topology, may also prove to be useful. Another possible extension is to the larger class of topological (and combinatorial) objects called hypergraphs. An additional promising route is to improve our formulas by taking into account that the probability of coalition formation may be non-binary.



\section{Weighted voting games and power indices}\label{S:WeightedGames}

In this section we will briefly recall the basics of simple games, weighted voting games, and the Banzhaf and Shapley-Shubik indices. This material and further details can be found in many standard references, such as  \cite{FM:VotingPower, FM:Misreinvention, LV:VotingDecisions, Napel:VotingPower, TP:MathPolitics, TZ:SimpleGames, W:ProbGames}.

\subsection{Simple and weighted games}\label{SS:SimpleGames}

Simple games model voting systems in which a binary choice is presented, or an alternative is presented against the status quo (selecting one of two candidates, passing a legislation, adopting a measure, convicting a suspect, etc.).

\begin{defin}\label{D:SimpleGame} A \emph{simple (voting) game} on a finite set $N$  is a function
$$
v\colon \mathcal P(N) \longrightarrow \{0,1\},
$$
where $\mathcal P(N)$ is the power set of $N$, satisfying:
\begin{enumerate}
\item $v(\emptyset)=0$;
\item $v(N)=1$;
\item If $S\subseteq T\subseteq N$ and $v(S)=1$, then $v(T)=1$ (\emph{monotonicity}). 
\end{enumerate}
\end{defin}

We will typically write $N=\{v_1, v_2, ..., v_n\}$. Elements $v_i\in N$ are called \emph{voters} (or \emph{players} or \emph{agents}). A subset $S$ of $N$ is called a \emph{coalition}. If $v(S)=1$, then $S$ is a \emph{winning} coalition. Otherwise it is a \emph{losing} coalition. A game is determined by its winning or its losing coalitions which we will denote by $\mathcal W$ and $\mathcal L$, respectively.

\begin{rem}\label{R:Hypergraph}
A simple game can also be regarded as a \emph{hypergraph} on $N$ -- which is by definition a collection of subsets of $N$ -- whose \emph{hyperedges} correspond to winning coalitions. In addition, we require that all supersets of hyperedges are also in the hypergraph because of condition (3) above and that the hypergraph contains $N$ itself as an edge because of condition (2) (which, incidentally, it will by (3) as soon as the hypergraph is non-empty since $N$ is a superset of any hyperedge). 
\end{rem}

A \emph{minimal winning coalition} is a subset $S$ such that $v(S)=1$ but the value of $v$ restricted to any proper subset of $S$ is 0, namely $S$ is a winning coalition but all its subsets are losing coalitions. Monotonicity implies that a simple game is determined by its minimal winning coalitions. Similarly, a \emph{maximal losing coalition} is one whose supersets are all winning. A game is also determined by its maximal losing coalitions.

Here is some standard terminology associated to simple games. We will not go into any depth about these definitions; suffice it to say that they are central to understanding simple games. Let $S^c$ denote the complement of $S$.

\begin{defin}\label{D:GameStuff}
\ 
\begin{itemize}
\item
A simple game is \emph{proper} if it satisfies: If $S$ is a winning coalition, then $S^c$ is not.
\item
A simple game is \emph{strong} if it satisfies: If $S$ is a losing coalition, then $S^c$ is not. 
 \item
A simple game is \emph{constant sum} if it is proper and strong. 
\item Voter $v_i$ is a \emph{dummy} if: For all $S\in \mathcal{W}$, $S\setminus \{v_i\}\in \mathcal{W}$.
\item Voter $v_i$ is a \emph{vetoer} if: Whenever $v_i\notin S$,  $S\notin \mathcal{W}$.
\item Voter $v_i$ is a \emph{passer} if: Whenever $v_i\in S$,  $S\in \mathcal{W}$.
\item Voter $v_i$ is a \emph{dictator} if:  $S\in \mathcal{W}$ if and only if $\{v_i\}\in S$. 
\item The \emph{dual game} $v^*$ of a game $v$ has the same voter set as $v$ and winning coalitions whose complements are losing coalitions of $v$.
\end{itemize}
\end{defin}

A simple game can also be defined on a subset $\mathcal F$ of $\mathcal P(N)$. An element $S\in \mathcal F$ is called a \emph{feasible coalition}. Otherwise it is \emph{unfeasible}. The idea is that, in real-life applications, not all coalitions can be formed; sometimes voters are not willing to align with some other voters for a particular issue at hand, or even under any circumstances at all. Our simplicial complex model will capture precisely this situation.

Suppose we have a \emph{weight function}
$$
w\colon N \longrightarrow \R_+.
$$
The value $w_i=w(v_i)$ is the \emph{weight} of the voter $v_i$, $i=1,...,n$. The heuristic is that voter $v_i$'s vote is ``worth'' $w_i$ votes, or carries $w_i$ points. Also suppose given a real number $0<q\leq t$, with $t=w_1+\cdots+w_n$, the total weight.

\begin{defin}\label{D:WeightedGame}
A simple game $v$ on $N$ is \emph{weighted} if there exists a weight function $w$ on $N$ and, for all $S\subseteq N$, 
$$
v(S)=1\ \  \Longleftrightarrow\ \  w(S)=\sum_{v_i\in S} w_i \geq q.
$$

\end{defin}

This definition says that a coalition $S$ is winning precisely when the total number of votes in $S$ clears the quota. The real-life analog is a collection of voters with potentially differing weights -- such as members of a board of directors whose voting weight is commensurate with the amount of stock they own or parties in a parliament which have different numbers of representatives -- forming a coalition and voting the same way on the choice that is presented. The coalition wins, or imposes their preference, if it contains enough votes to surpass the quota $q$, typically set as $\frac{1}{2}t<q \leq t$ (the usual \emph{supermajority} condition). In this situation, $S$ and $N\setminus S$ cannot both be winning so the game is proper.

We will denote a weighted voting game on $N=\{v_1, v_2, ..., v_n\}$ with weights $w_i$, $1\leq i \leq n$, and quota $q$ by $V(q; w_1,...,w_n)$; this is a standard shorthand notation.


Not every simple game is weighted. The standard examples are that of the U.S.~legislative system and the procedure to amend the Canadian Constitution \cite{J:GameCategory, TZ:SimpleGames}.

If all voters have the same weight, i.e.~$w_i=m$ for all $1\leq i\leq n$, then we can set the weights at 1 and scale the quota by $m$. This produces an isomorphic weighted system, namely the value of the map $v$ is the same on all subsets of $N$ before and after scaling. We can therefore assume that, for weighted games where 
all weights are equal, those weights are in fact 1. We will denote such a weighted system by $V(q; 1,...,1)$ and call it \emph{symmetric} (it is also known as \emph{$q$-out-of-$n$} in the literature).

%
%
%
%
%
%


\subsection{Banzhaf index}\label{SS:Banzhaf} The theory of power indices arose from the simple observation that a voter's weight does not paint the entire picture in terms of that voter's ability to influence outcomes. For example, in the weighted game $V(51; 49, 49, 2)$, all three voters appear the same number of times in minimal winning coalitions so in an intuitive sense have the same importance. 

Banzhaf index is one of the measures that tries to capture this discrepancy. It was defined, in slightly different forms, by Penrose \cite{Pen46} and Banzhaf \cite{Ban65} (see also Coleman \cite{Col71}). It is also sometimes referred to as the Penrose-Banzhaf or Banzhaf-Coleman index. 

\begin{defin}
A voter $v_i$ is said to be \emph{critical} for a coalition $S$ in a simple game $v$ if $S$ is a winning coalition but $S\setminus \{v_i\}$ is a losing coalition.
\end{defin}

In case of a weighted system $V(q; w_1,...,w_n)$, this means that $v_i$ is critical for $S$ if both these inequalities hold:
\begin{align*}
\sum_{v_j\in S}w_j\geq q; \\
 \sum_{v_j\in S\setminus \{v_i\}}w_j<q.
 \end{align*}

\begin{defin}\label{D:BanzhafIndex}
The \emph{absolute Banzhaf (power) index} of the voter $v_i$ in a simple game $v$ is defined to be 
\begin{equation}\label{E:AbsoluteBanzhaf}
AB(v_i)=\frac{1}{2^{n-1}}\sum_{S\subseteq N\setminus\{v_i\}} \big(v(S\cup \{v_i\}) - v(S)  \big).
\end{equation}
\end{defin}
This index can be interpreted as the probability that voter $v_i$ is critical assuming all coalitions are equally likely. The absolute Banzhaf index has some pleasant formal properties and is in fact determined by them.

One inconvenience with the absolute Banzhaf index is that its values do not add up to 1. This is why it is sometimes convenient to use the \emph{relative} or \emph{normalized} Banzhaf index, given as
\begin{equation}\label{E:RelativeBanzhaf}
B(v_i)=\frac{AB(v_i)}{\sum_{v_i\in N}AB(v_i)}.
\end{equation}
In other words, $B(v_i)$ is the number of times voter $v_i$ is critical divided by the total number of times all voters are critical; thus, if we denote by $c(v_i)$ the number of times voter $v_i$ is critical and by $c(V)$ the total number of times all voters are critical, then the Banzhaf index can be written as $$B(v_i)=\frac{c(v_i)}{c(V)}.$$
By design, $\sum_{v_i\in N}B(v_i)=1$ so that each voter's power can be expressed as a percentage of total power. The shift between the two indices is just a matter of scaling. 

In our situation, it additionally makes sense to use the relative Banzhaf index since we will not consider all coalitions equally likely. In fact, some of them will appear with probability zero, and this is precisely what we want to capture with simplicial complexes. Axiomatic features of the absolute Banzhaf index will also not be of use to us since they also assume all coalitions are present. For example, the \emph{monotonicity} axiom -- which says that if $S$ is a winning coalition and  $S\subset T$, then $T$ is also winning -- no longer makes sense if we remove some coalitions. We will thus use the relative version and will refer to it simply as the Banzhaf index.

Here is an example using a weighted voting game, the situation we will generally be concerned with.

\begin{example}\label{E:Banzhaf1}
Take the weighted voting game $V(q=8; w_1=1, w_2=2, w_3=3, w_4=4, w_5=5)$ of five voters $\{v_1, v_2, v_3, v_4, v_5\}$. There are $12$ winning coalitions, including the coalition $\{v_1, v_2, v_3, v_4\}$. This is a winning coalition because $1+2+3+4=10 \geq 8$. In this coalition, only $v_3$ and $v_4$ are critical, because without either of their votes the coalition would not pass, but $v_1$ and $v_2$ could each be removed and the coalition would still be winning. Out of the $12$ winning coalitions, there are $26$ instances when a voter is critical. Voter $v_1$ is critical $2$ times, namely in the winning coalitions $\{v_1, v_2, v_5\}$ and $\{v_1, v_3, v_4\}$, so $c(v_1)=2$. We also have $c(v_2)=2$ because $v_2$ is critical only in the winning coalitions $\{v_1, v_2, v_5\}$ and $\{v_2, v_3, v_4\}$. Similarly, $c(v_3)=6, c(v_4)=6$, and $c(v_5)=10$. Thus the Banzhaf indices for this weighted game are
$$B(v_1)=B(v_2)=\frac{2}{26},\ \  B(v_3)=B(v_4)=\frac{6}{26},\ \  B(v_5)=\frac{10}{26} .$$
\end{example}


\subsection{Shapley-Shubik index}\label{SS:ShSh}
While for the Banzhaf index the order in which voters join a coalition does not matter, i.e.~the coalitions are just subsets of the set of voters, the Shapley-Shubik index, introduced by Shapley and Shubik in 1954 \cite{ShSh54} takes the order in which voters enter a coalition into account. This is natural since, in real-life situations, voters often decide whether to join a coalition based on who is already in it (think political parties in a parliament). Voting power is then measured by how often a voter flips a coalition from losing to winning at the moment they join. 

Let $N=\{v_1, ..., v_n\}$ be the set of voters as before.

\begin{defin}
The \emph{Shapley-Shubik (power) index} of a voter $v_i$ in a simple game $v$  is defined to be 

\begin{equation}\label{E:SS}
SS(v_i)=\frac{1}{n!}\sum_{S\subseteq N\setminus \{v_i\}} (|S|!\cdot(n-|S|-1)!\big(v(S\cup \{v_i\}) - v(S)  \big).
\end{equation}

\end{defin}

To see the relationship to the idea that voters might be joining coalitions in sequence, here is an alternative description of $SS(v_i)$.

\begin{defin}
A \emph{sequential coalition} (or an \emph{ordered coalition}) of $N$ is any bijection $\sigma: N \rightarrow N$, that is, any permutation $[v_{\sigma(1)}, ..., v_{\sigma(n)}]$ of $N$. Let $Sym(N)$ be the set of all permutations of $N$.
\end{defin}

For $\sigma\in Sym(N)$, let $S_{\sigma(<k)}$ be the truncation of the ordered coalition $[v_{\sigma(1)}, ..., v_{\sigma(n)}]$ at $v_{\sigma(k)}$ (but not including $v_{\sigma(k)}$).

Regarding $S_{\sigma(<k)}$ as a set, we can evaluate a simple game $v$ on it and determine whether it is a winning or a losing coalition.

\begin{defin} Let $\sigma$ be an ordered coalition and let $i=\sigma(k)$ for some $k\in \{1,...,n\}$. 
A voter $v_i$ is said to be \emph{pivotal} for $\sigma$ if $S_{\sigma(<k)}$ is losing but $S_{\sigma(<k)}\cup \{v_i\}$ is winning.
\end{defin}

In other words, $v_i$ is pivotal for a sequential coalition if the coalition is losing until $v_i$ joins it. Then an alternative expression for $SS(v_i)$ is given by

\begin{equation}\label{E:ShapleyShubik}
SS(v_i)=\frac{1}{n!}\sum_{\sigma\in Sym(N)} \big(v(S_{\sigma(<\sigma^{-1}(i))}\cup \{v_i\}) - v(S_{\sigma(<\sigma^{-1}(i))})  \big).
\end{equation}

This formula simply counts the number of times $v_i$ is pivotal and scales by the number of all possible sequential coalitions, namely $n!$. Thus is we denote by $p(v_i)$ the number of instances $v_i$ is pivotal, the Shapley-Shubik index can be rewritten a little more succinctly as 
$$SS(v_i)=\frac{p(v_i)}{n!}.$$
%

Assuming all sequential coalitions are equally likely, this gives the probability of being pivotal. Since every sequential coalition has a pivotal voter, it follows that $\sum_{v_i\in N}SS(v_i)=1$. The normalization condition (called \emph{efficiency} in the literature) is thus built into the definition of the Shapley-Shubik index.

For an explanation of how the formulas for the Shapley-Shubik index from Equations \ref{E:SS} and \ref{E:ShapleyShubik} are related, see \cite{Napel:VotingPower}.

In the special case of  a weighted game, we can say more about what it means to be a pivotal voter: $v_i$ is pivotal for $[v_{\sigma(1)}, ..., v_{\sigma(n)}]$ if $i=\sigma(k)$ for some $k\in \{1,...,n\}$ and these inequalities hold:
\begin{align*}
\sum_{1\leq j\leq k}w_{\sigma(j)}\geq q; \\
\sum_{1\leq j< k}w_{\sigma(j)}<q.
%
%
\end{align*}

\begin{example}
Consider the weighted game $V(8; 1, 2, 3, 4, 5)$ as in Example \ref{E:Banzhaf1}. There are $5!=120$ sequential coalitions. For example, two of them are $A=[v_5, v_1, v_2, v_4, v_3]$ and $B=[v_1, v_3, v_2, v_4, v_5]$. In $A$, $v_2$ is the pivotal voter, because $v_5$ and $v_1$ together do not have enough voting power, but $v_2$ add enough votes to make this coalition into a winning one. Similarly, in the sequential coalition $B$, $v_4$ is the pivotal voter. We find that all sequential coalitions with $v_1$ as pivotal are $[v_2, v_5, v_1, v_4, v_3]$, $[v_2, v_5, v_1, v_3, v_4]$, $[v_5, v_2, v_1, v_3, v_4]$, $[v_5, v_2, v_1, v_4, v_3]$, $[v_3, v_4, v_1, v_2, v_5]$, $[v_3, v_4, v_1, v_5, v_2]$, $[v_4, v_3, v_1, v_2, v_5]$, $[v_4, v_3, v_1, v_5, v_2]$. Hence, $p(v_1)=8$. Similarly, $p(v_2)=8$, $p(v_3)=p(v_4)=28$, $p(v_5)=48$. Therefore, the Shapley-Shubik indices for this weighted game are
$$SS(v_1)=SS(v_2)=\frac{8}{120},\ \  SS(v_3)=SS(v_4)=\frac{28}{120},\ \  SS(v_5)=\frac{48}{120}.$$
Comparing this with Example \ref{E:Banzhaf1}, we see that the Shapley-Shubik indices for the voters $v_1$ and $v_2$ are slightly less than the corresponding Banzhaf indices, and the Shapley-Shubik indices for the voters $v_3, v_4, v_5$ are slightly greater than the corresponding Banzhaf indices.
\end{example}

\section{Simplicial complexes}\label{S:SC}
This section provides the basic background on simplicial complexes that we will use to represent voting games and develop new formulas for power indices. Simplicial complexes are standard in topology; more about them can be found in many sources, such as \cite{FPR09, Mat03, Spa95}. A self-contained exposition sufficient for understanding how political structures relate to simplicial complexes can be found in \cite{MV:Politics}.

\begin{defin}\label{D:SC}
A \emph{finite (abstract) simplicial complex} is the ordered pair $K=(V,\Delta)$ that consists of a finite set $V$ and a set $\Delta$ of subsets of $V$ such that\\
\hspace*{1cm}(1) Elements of $V$ are in $\Delta$;\\
\hspace*{1cm}(2) If $\sigma \in \Delta$ and $\eta \subset \sigma$, then $\eta \in \Delta$.
\end{defin}
Roughly speaking, $K$ is a family of sets that is closed under taking subsets. Elements of $V$ are called \textit{vertices} of $K$, and elements of $\Delta$ are called \textit{simplices} of $K$.

To simplify notation, we usually identify $K$ with $\Delta$.

If $V=\{v_1, ..., v_n\}$ and $K=\mathcal P_0(V)$, where this denotes the set of all non-empty subsets of $V$, then $K$ is called the \emph{standard $(n-1)$-simplex} and denoted by $\Delta^{n-1}$.

One can associate to a simplicial complex $K$ a topological space $|\!|K|\!|$ called the \emph{geometric realization} (or the \emph{polyhedron}) of $K$, obtained as a union of points, line segments, triangles, tetrahedra, and their higher-dimensional generalizations. Each point represents a vertex and a (generalized) tetrahedron spanned by some subsets of the vertices represents a simplex. We think of $|\!|K|\!|$ as embedded in some Euclidean space of sufficiently large dimension. When we say \textit{simplicial complex}, we will mean both the abstract complex and its geometric realization, and will denote both by $K$; this will not create confusion.

\begin{example}\label{E:complex1}
An example of a simplicial complex on the vertex set $V=\{v_1, v_2, v_3, v_4, v_5, v_6, v_7\}$ is 
$$K=\{\{v_1\}, \{v_2\}, \{v_3\}, \{v_4\}, \{v_5\}, \{v_6\}, \{v_7\}, \{v_1, v_2\}, \{v_1, v_3\},$$
$$\hspace{0.78cm}\{v_2, v_3\}, \{v_2, v_4\}, \{v_2, v_5\}, \{v_4, v_5\}, \{v_5, v_6\}, \{v_1, v_2, v_3\}\}.$$
The geometric realization of $K$ is given in Figure \ref{F:SimpComplex1}.

On the other hand, the collections 
$$K'=\{\{v_1\}, \{v_2\}, \{v_3\}, \{v_4\}, \{v_5\}, \{v_6\}, \{v_7\}, \{v_1, v_2\}, \{v_1, v_3\}, \{v_1, v_2, v_3\}\}$$
and
$$K''=\{\{v_1\}, \{v_2\}, \{v_3\}, \{v_4\}, \{v_5\}, \{v_6\}, \{v_1, v_2\}, \{v_1, v_3\}, $$
$$\hspace{1.78cm}\{v_2, v_3\}, \{v_2, v_4\}, \{v_2, v_5\}, \{v_4, v_5\}, \{v_5, v_6\}, \{v_1, v_2, v_3\}\}$$
are not simplicial complexes on $V$; $K'$ does not contain the subset $\{v_2,v_3\}$ as required by part (2) of \refD{SC}, and $K''$ does not contain $\{v_7\}$ as required by part (1) of \refD{SC}.
\end{example}

\begin{figure}
\begin{center}
\begin{tikzpicture}
\begin{scope}
\fill[blue!20] (0,0) -- (4,0) -- (2, 2.82);
\end{scope}
\draw[fill] (0,0) circle [radius=0.1]; \draw[fill] (4,0) circle [radius=0.1]; \draw[fill] (2,2.82) circle [radius=0.1]; \draw[fill] (8,0) circle [radius=0.1]; \draw[fill] (6,2.82) circle [radius=0.1]; \draw[fill] (10,2.82) circle [radius=0.1]; \draw[fill] (12,0) circle [radius=0.1];
\draw (0,0) node[below] {$v_1$} -- (4,0) node [below] {$v_2$} -- (2,2.82) node[above] {$v_3$} -- (0,0) node[below] {$v_1$}; 
\draw  (4,0) node [below] {$v_2$} -- (8,0) node [below] {$v_5$}; \draw  (4,0) node [below] {$v_2$} -- (6,2.82) node [above] {$v_4$};  \draw  (8,0) node [below] {$v_5$} -- (6,2.82) node [above] {$v_4$}; \draw  (8,0) node [below] {$v_5$} -- (10,2.82) node [above] {$v_6$}; \draw (12,0) node[below] {$v_7$};
\end{tikzpicture}
\end{center}
\caption{The geometric realization of the simplicial complex from Example \ref{E:complex1}}.
\label{F:SimpComplex1}
\end{figure}
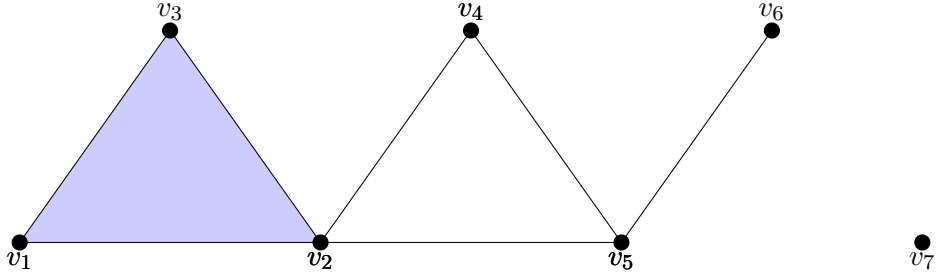

Here is some standard terminology associated to simplicial complexes that we will need.

\begin{defin}\label{D:SCsStuff}\ 

\begin{itemize}
\item A simplicial complex $L=(V', \Delta')$ is a \emph{subcomplex} of a simplicial complex $K=(V,\Delta)$ if $V'\subset V$ and $\Delta' \subset \Delta$. 
\item Let $K=(V, \Delta)$ be a simplicial complex. We call $\sigma \in \Delta$ a \emph{$d$-simplex} of $K$ if $\sigma$ contains exactly $d+1$ vertices, i.e. $|\sigma|=d+1$.
\item If $\sigma$ is a simplex of $K$, then any subset $\eta$ of $\sigma$ is called a \emph{face} of $\sigma$.
\item If $\sigma$ is a simplex of $K$, the \emph{dimension} of $\sigma$ is $$\dim \sigma=|\sigma|-1.$$
\item The \emph{dimension of the complex} $K$ is defined to be the maximal dimension of its simplices, that is
$$\dim K=\max \{\dim \sigma: \sigma \in \Delta\}.$$
If $\dim K=n$, we say $K$ is an \emph{$n$-complex}.
\item A simplex $\sigma$ in $K$ is called a \emph{maximal simplex} or a \emph{facet} if $\sigma$ is not a face of any other simplex of $K$. It follows that $\Delta$ is the set of all non-empty subsets of the maximal simplices of $K$.
\item Let $K^d$ be the collection of all $e$-simplices of $K$ along with their faces, for all $0 \leq e \leq d$. Then $K^d$ is called the \emph{$d$-skeleton} of $K$. Note that $K^l \subset K^m$ for $l \le m$.
\item The \emph{star} of a simplex $\sigma$ is defined to be the set $$\text{st}(\sigma)=\{\tau \in \Delta: \sigma \subset \tau\}.$$
Here we are most interested in the special case $\sigma=\{v\}$ for a vertex $v$, so $\text{st}(v)$ is the set of all simplices in $K$ that contain $v$ as a vertex.
\item The action of removing the set $\text{st}(v)$ from the complex $K$ is called the \emph{deletion} of $v$. The subcomplex $K\setminus \text{st}(v)$ is typically denoted simply by $K\setminus v$.
\item A \emph{cone} on $K$ is the simplicial complex with simplices $K\cup \{c\} \cup \{\sigma\cup \{c\}\colon \sigma\in K\}$.
\item The \emph{Alexander dual} of $K$ is the simplicial complex with simplices $\{\sigma\in \mathcal P(V)\colon \sigma^c \notin K\}$.
\end{itemize}
\end{defin}

The cone and the Alexander dual are important constructions from algebraic topology; they will make an appearance in \refP{Games<->SCs}. The rest of the definitions will be needed throughout, and especially in Section \ref{S:Computing}.

\begin{example}\label{E:complex2}
Let $K$ be the simplicial complex from Example \ref{E:complex1}, with the geometric realization presented in Figure \ref{F:SimpComplex1}.
\begin{itemize}
\item Each vertex $\{v_i\}$, $i=1, ..., 7$, is $0$-simplex of $K$.\\
The line segments $\{v_1, v_2\}, \{v_1, v_3\}, \{v_2, v_3\}, \{v_2, v_4\}, \{v_2, v_5\}, \{v_4, v_5\}, \{v_5, v_6\}$ are $1$-simplices of $K$.\\
$K$ has one $2$-simplex, a triangle $\{v_1, v_2, v_3\}$.
\item The dimension of $K$ is $2$.
\item The simplices $\{v_1, v_2\}, \{v_1, v_3\}, \{v_2, v_3\}$ are faces of the simplex $\{v_1, v_2, v_3\}$.
\item The simplices $\{v_1, v_2, v_3\}$, $\{v_2, v_4\}, \{v_2, v_5\}, \{v_4, v_5\}, \{v_5, v_6\}$ and $\{v_7\}$ are maximal simplices of $K$. The simplices $\{v_1, v_2\}, \{v_1, v_3\}, \{v_2, v_3\}$ are not maximal.
\item $K^0=\{\{v_1\}, \{v_2\}, \{v_3\}, \{v_4\}, \{v_5\},\{v_6\}, \{v_7\}\}.$\\
$K^1=K^0\cup \{\{v_1, v_2\}, \{v_1, v_3\}, \{v_2, v_3\}, \{v_2, v_4\}, \{v_2, v_5\}, \{v_4, v_5\}, \{v_5, v_6\}\}.$\\
$K^2=K^1 \cup \{\{v_1, v_2, v_3\}\}.$
\item $\text{st}(v_2)=\{\{v_2\}, \{v_1, v_2\}, \{v_2, v_3\}, \{v_2, v_4\}, \{v_2, v_5\}, \{v_1, v_2, v_3\}\}$.\\
$K\setminus v_2 = \{\{v_1\}, \{v_3\}, \{v_4\}, \{v_5\}, \{v_6\}, \{v_7\}, \{v_1, v_3\}, \{v_4, v_5\}, \{v_5, v_6\}\}$.
\end{itemize}
\end{example}


Let $V=\{v_1, ..., v_n\}$ be the vertex set as before.

\begin{defin}\label{D:Complete} Let $0\leq d< n$. 
A \emph{complete} simplicial complex on $n$ vertices of dimension $d$, denoted by $K(n,d)$, is the complex that contains all subsets of the vertex set of cardinality less than or equal to $d+1$ (and no subsets of greater cardinality).
\end{defin}
In other words, $K(n,d)$ contains all possible faces of dimension $d$ or less. 
If $d=2$, $K(n,2)$ produces a complete graph on $n$ vertices. If $d=n-1$, then this produces the standard $(n-1)$-simplex $\Delta^{n-1}$. If $d=n-2$, the complex is homeomorphic to the $(d-1)$-sphere; for example, $K(3,1)$ consists of three edges forming a triangle which is topologically the circle $S^1$ while $K(4,2)$ consists of four triangles fitting together to form a  hollow tetrahedron, which is topologically the 2-sphere $S^2$.

One final notion we need is that of the \emph{homology} of $K$. Defining it here precisely would take us too far afield, but details can be found in \cite{MV:Politics, VV:SimpleGames} (with the first reference providing a more thorough treatment). Suffice it to say that homology groups $H_i(K)$, $i\geq 0$, keep track of the number of $i$-dimensional holes, or voids in $K$. Each homology group is a product of some number of copies of the integers $\Z$, and the number of copies is the $i$th \emph{Betti number}, $\beta_i(K)$. The zeroth Betti number indicates the number of connected components of $K$.

For example, if $K=\Delta^{n-1}$, all Betti numbers are zero since the simplex has no holes. If we take out $V=\{v_1, ..., v_n\}$ itself out of $K$, then the Betti numbers remain zero except for $\beta_{n-2}(K)=1$ (e.g.~taking out the interior of $\Delta^3$ leaves a hollow tetrahedron which now has one 2-dimensional hole). 

Here is an observation which will be useful later.

\begin{prop}\label{P:CompleteHomology}
A simplicial complex $K$ is a complete complex $K(n,d)$ if and only if its homology is given by
$$
\beta_i (K(n,d)) =
\begin{cases}
{n \choose d+2}, & i=d ;\\
0, & \text{otherwise}.
\end{cases}
$$

\end{prop}

\begin{proof}
We give an idea of the proof. A more complete argument can be extracted from \cite[Theorem 4.5]{VV:SimpleGames} with just a slight modification in the indexing scheme.

If a complex is $K(n,d)$, this means that every face of dimension $d$ appears. In order to form a homology class, $d+2$ such faces are needed (3 edges to form a 1-dimensional cycle or a triangle, 4 triangles to form a 2-dimensional cycle or a hollow tetrahedron, etc.). Each choice of $d+2$ vertices out of $n$ will produce such a homology class, which gives the Betti number of ${n \choose d+2}$. No other homology is present since every lower-dimensional cycle is ``capped-off'' by a higher-dimensional face (and is hence homologically trivial) and there are no higher homology groups since there are no simplices of dimension greater than $d$.

For the reverse direction, if $\beta_d(K)={n \choose d+2}$, this means that there are ${n \choose d+2}$ $d$-dimensional homology cycles, which in turn means that every possible face of dimension $d$ is present (and lower-dimensional ones since this is a simplicial complex); the last statement is true because the minimum number of vertices required to form a $d$-cycle is $d+2$. There cannot be simplices of dimension higher than $d$ since a face of dimension $d+1$ would cap off a $d$-dimensional homology class and reduce the Betti number, and if there are no faces of dimension $d+1$, then there cannot be faces of dimensions higher than that (again because this is a simplicial complex).
\end{proof}


\section{Simplicial complex models for simple games}\label{S:SimplicialSimple}


Recall that losing coalitions of a simple game determine the game; if we know the losing coalitions, then we know that those that are not losing are winning, so we thus have all the information about the function $v\colon \mathcal P\to\{0,1\}$. Furthermore, the monotonicity condition (3) of \refD{SimpleGame} translates into the following condition for losing coalitions:
\begin{equation}\label{E:LosingCondition}
\text{If $T\subset S\subseteq N$ and $v(S)=0$, then $v(T)=0$}.
\end{equation}

Heuristicaly, if a coalition is losing, it does not have enough votes to push their agenda through, and if that is the case, no smaller subset of the voters will have the power to do so either.

It should be clear from \eqref{E:LosingCondition}, in comparison to \refD{SC}, that the losing coalitions $\mathcal L$ form a simplicial complex which we will denote by $K_{\mathcal L}$:
\begin{align*}
\text{voters}\  v_i\in N\ \ & \longleftrightarrow\ \    \text{vertices}\  v_i\in K_{\mathcal L} \\
\text{losing coalitions}\  S\in \mathcal L \ \ & \longleftrightarrow\ \  \text{simplices}\  \sigma\in K_{\mathcal L}
\end{align*}

With this correspondence, it is easy to establish a dictionary between various notions in simple games and those in simplicial complexes. In particular, the proofs of the following are easy consequences of the definitions (and can be found in \cite[Propositions 4.1 and 4.4]{VV:SimpleGames}). Recall the definition of a cone and the Alexander dual from \refD{SCsStuff}.

\begin{prop}\label{P:Games<->SCs} \ 
\begin{itemize}
\item If the game has a dummy, then $K_{\mathcal L}$ is a cone on the simplicial complex of the losing coalitions of the same game but with the dummy removed from the voter set.
\item If the game has a passer, then $K_{\mathcal L}$ does not contain the passer as a vertex.
\item If the game has a dictator, then $K_{\mathcal L}$ is the $(n-2)$-simplex on all vertices except the dictator.
\item Dual game $\mathcal L^*$ corresponds to the Alexander dual $K_{\mathcal L}^*$, and a game is constant sum if and only if the simplicial complex is self-dual, i.e.~$K_{\mathcal L}=K_{\mathcal L}^*$. 
\end{itemize}
\end{prop}

The following theorem and its consequence are the main new results of this section, illustrating the potential utility of studying simple games from a topological point of view. Recall the definition of a complete complex from \refD{Complete}.

\begin{thm}\label{T:CompleteWeightedGames}
Suppose $2\leq q\leq k+1$. 
A simple game with $n$ voters is isomorphic to the weighted game
\begin{equation}\label{E:WeightedComplete}
V(q;\underbrace{1,1,\dots,1}_{k},\underbrace{q,q,\dots,q}_{n-k})
\end{equation}
if and only if the simplicial complex of losing coalitions $K_{\mathcal L}$ is the complete complex $K_{\mathcal L}(k,q-2)$.
\end{thm}

Note that voters $\{v_{k+1}, ..., v_n\}$ are passers. Also note that, if $q>k+1$, then the game is isomorphic to the one where $q$ is set to equal $k+1$. For example, $V(5;1,1,1,5,5)$ is isomorphic to $V(4;1,1,1,4,4)$. In this situation, the complex is $K_{\mathcal L}(k,k-1)$, which corresponds to the standard $(k-1)$-simplex $\Delta^{k-1}$.

\begin{proof}
First assume a weighted game as in \eqref{E:WeightedComplete} is given. Assume voters $v_1, ..., v_k$ are the ones with weight 1. Then any collection of $q-1$ or fewer of them forms a losing coalition. Since none of the voters $v_{k+1}, ..., v_n$ can be in a losing coalition (they clear the quota individually, so they form singleton winning coalitions, and hence any coalition containing them is also winning), this is all the faces that are present in $K_{\mathcal L}$. But those faces form precisely the complete complex of dimension $q-2$ on $k$ vertices.

Conversely, suppose the losing coalitions form the complete complex $K_{\mathcal L}(k,q-2)$ and again suppose without loss of generality that the vertices correspond to voters $v_1, ..., v_k$. It is clear that a game that has this complex of losing coalitions is $V(q;1,1,...,1)$, where 1 appears $k$ times. But since we know that the voter set has size $n$, and the remaining voters $v_{k+1}, ..., v_n$ are not in the losing complex, they individually must form singleton winning coalitions. Thus a weight of $q$ can be assigned to them (or anything greater than $q$, but $q$ is the smallest value), so that the final weighted game is precisely the one from \eqref{E:WeightedComplete}.
\end{proof}

From \refP{CompleteHomology}, we immediately have

\begin{cor}\label{C:CompleteHomologyWeighted}
With the same assumptions as in \refT{CompleteWeightedGames}, a game is isomorphic to the one in \eqref{E:WeightedComplete} if and only if the homology of the simplicial complex of the losing coalitions has rank ${k \choose q}$ in dimension $q-2$ and zero otherwise.
\end{cor}

This result generalizes \cite[Theorem 4.5]{VV:SimpleGames}. In that situation $n=k$, and the result says that the homology of the losing coalitions has rank ${n \choose q}$ in dimension $q-2$ and zero otherwise if an only if the simple game is symmetric, i.e.~of the form $V(q;1,1...,1)$.

The hope is that this result can be generalized to arbitrary weighted games, establishing their characterization in terms of the topology of the simplicial complex of losing coalitions.


\section{Simplicial complex formulas for power indices of games with unfeasible coalitions}\label{S:Computing}


We now turn to another way simplicial complexes make an appearance in the study of simple games. Namely, we use them to model weighted voting games where some of the coalitions are not possible and calculate the Banzhaf and the Shapley-Shubik indices for such systems.


Consider first the traditional game in which no coalition is assumed unfeasible. Let $v$ as usual be a simple game on the set of voters $N=\{v_1, ..., v_n\}$. Then we can represent the domain of $v$ as the standard $(n-1)$-simplex $\Delta^{n-1}$ with vertex set $N$ and with faces representing coalitions:
\begin{align*}
\text{voters}\  v_i\in N\ \ & \longleftrightarrow\ \    \text{vertices}\  v_i\in \Delta^{n-1} \\
\text{(nonempty) coalitions}\  S\in \mathcal{P}_0(N) \ \ & \longleftrightarrow\ \  \text{simplices}\  \sigma\in \Delta^{n-1}
\end{align*}
The simple game is then a map
$$
v\colon \Delta^{n-1}\longrightarrow \{0,1\}
$$
which is extended to the empty set as $v(\emptyset)=0$.

If a game has unfeasible coalitions, this means that $v$ might only be defined on $\mathcal F\subset \mathcal P(N)$ consisting of subsets of feasible coalitions. For a weighted game, the voter weights determine the total weights of all possible (sequential) coalitions, so a game $v$ is automatically defined on all of $\mathcal P(N)$. Incorporating unfeasibilities into a weighted game thus means restricting $v$ to a function on $\mathcal F$, namely 
$$
v\colon \mathcal P(N)\vert_{\mathcal F} \longrightarrow \{0,1\}.
$$
From the simplicial complex point of view, incorporating unfeasible coalitions reduces to the following: Start with $\Delta^{n-1}$. For each unfeasible coalition, remove the face corresponding to it. Also remove all the faces containing that face, the assumption being that, if a set of voters is incompatible, any larger set containing them is also incompatible. This leaves a simplicial complex $K$:
\begin{align*}
\text{voters}\  v_i\in N\ \ & \longleftrightarrow\ \    \text{vertices}\  v_i\in K\\
\text{(nonempty) feasible coalitions}\  S\in \mathcal F \ \ & \longleftrightarrow\ \  \text{simplices}\  \sigma\in K
\end{align*}
A simple game is now a map
$$
v\colon K\longrightarrow \{0,1\}, \ \ v(\emptyset)=0.
$$
Alternatively, $v$ can be represented as a \emph{labeled} simplicial complex with each face labeled by its image under $v$. We will denote the label value on a face $\sigma$ by $l(\sigma)$. Thus for a general simple game
$$
l(\sigma) =1 \Longleftrightarrow v(S)=1,
$$
where $\sigma$ is the face corresponding to the coalition $S$. Similarly for $l(\sigma)=0$. So if a face carries the label 1, it represents a winning coalition. If it is labeled with 0, it stands for a losing coalition. 

For weighted games, this labeling can be replaced with labels on vertices indicated each voter's weight. This induces a labeling on each face, given by the sum of the weights of the vertices spanning it. Recalling the notation from \refD{WeightedGame}, this therefore means
$$
l(\sigma) =1 \Longleftrightarrow w(S)\geq q.
$$
 We will refer to a simplicial complex with vertices labeled by voter weights and the induced labeling on faces as a \emph{weighted} simplicial complex.



\begin{example}
Let $v$ be the voting game on $N=\{v_1, ..., v_7\}$ such that a voter in $\{v_1, v_2, v_3\}$ could never be in a coalition with a voter in $\{v_5, v_6, v_7\}$. This means that certain subsets of the set $\{v_1,..., v_7\}$ are not allowed, such as $\{v_3, v_5\}$, $\{v_1, v_3, v_7\}$, $\{v_1, v_4, v_5\}$, and $\{v_2, v_5, v_6\}$. Therefore we remove all such simplices from the full $6$-simplex. On the other hand, the elements of $\{v_1, v_2, v_3\}$ and $\{v_5, v_6, v_7\}$ can be in a coalition with any other element of those sets, and $v_4$ is open to being in a coalition with anyone. The domain of the game can thus be represented as the simplicial complex $K$ given by
$$K=\mathcal{P}_0(\{v_1, v_2, v_3, v_4\})\cup \mathcal{P}_0(\{v_4, v_5, v_6, v_7\}).$$
The geometric realization of $K$ is the union of two tetrahedra, one with vertex set  $\{v_1, v_2, v_3, v_4\}$ and the other with vertices $\{v_4, v_5, v_6, v_7\}$, with the common vertex $v_4$; in other words, this is a wedge sum of two tetrahedra with $v_4$ as the wedge point.
\end{example}

\begin{example}\label{E:ElectoralC}
The president of the United States of America is not elected directly by popular vote, but through a process called \textit{Electoral College}. Each of the 50 states, along with District of Columbia, gets a certain number of electors, which depends on its population according to the latest census. The number of electors determined on the basis of 2020 Census will be applied to the 2024 and 2028 presidential elections. The smallest number of electors is 3 (for Alaska, Delaware, North Dakota, South Dakota, Vermont, and D.C.), and the highest number is 54 (for California). The total number of electors is 538, and a candidate needs at least 270 votes in the College to win the election.

The U.S.~presidential election can thus be represented as the weighted game 
\begin{align*}
V(270; 54, 40, 30, 28, 19, 19, 17, 16, 16, 15, 14, 13, 12, 11, 11, 11, 11, 10, 10, 10, \\10,
10, 9, 9, 8, 8, 8, 7, 7, 6, 6, 6, 6, 6, 6, 5, 5, 4, 4, 4, 4, 4, 4, 4, 3, 3, 3, 3, 3, 3, 3).
\end{align*}

The states are often divided into three groups: "blue wall" states, "red wall" states, and swing (or ``purple'') states. These are not rigorously defined categories, as the definition depends on the criteria used (usually along the lines of how many elections of what kind have been won by which party in each state). We take one of the possibilities here, but the analysis and calculations can easily be modified for any of them. Namely, we will think of the Electoral College as partitioned into the following 18 blue wall, 21 red wall, and 12 swing states:

\begin{tabular}{{p{.5in}p{5.5in}}}
Blue: & California, Colorado, Connecticut, Delaware, Hawaii, Illinois, Maryland, Massachusetts, Nevada, New Jersey, New Mexico, New York, Oregon, Rhode Island, Vermont, Virginia, Washington, District of Columbia. \\
Red: & Alabama,  Alaska, Arkansas, Idaho, Indiana, Kansas, Kentucky, Louisiana, Mississippi, Missouri, Montana, Nebraska, North Dakota, Oklahoma, South Carolina, South Dakota, Tennessee, Texas, Utah, West Virginia, Wyoming. \\
Swing: & Arizona, Florida, Georgia, Iowa, Maine, Michigan, Minnesota, New Hampshire, North Carolina, Ohio, Pennsylvania, Wisconsin.
\end{tabular}

%
%
%

In all states except Nebraska and Maine, the electors are assigned according to the ``winner-take-all'' principle. Nebraska and Maine allocate electoral votes districtwise, which means that there is potential for each of those states to not vote entirely blue or red, but we will ignore this possiblity since it is not of significant relevance for the calculations.

U.S.~presidential elections typically come down to a vote between one Republican candidate and one Democratic candidate. It follows that a “red wall” state is not going  to be in the same coalition as a “blue wall” state and consideration of such impossible coalitions when counting pivotal and critical voters may lead to misinterpretation of the situation on the ground. With this in mind, and reordering the states so that the first $18$ states are blue and the next $21$ states are red, we will represent the Electoral College as the simplicial complex $K$ given by 
$$K = \mathcal{P}_0(\{v_1,...,v_{18},v_{40},...,v_{51}\}) \cup \mathcal{P}_0(\{v_{19},...,v_{39},v_{40},...,v_{51}\})$$
where the vertices corresponding to the “blue wall” states are \{$v_1$,...,$v_{18}$\}, the vertices corresponding to the “red wall” states are \{$v_{19}$,...,$v_{39}$\}, and the vertices corresponding to the swing states are \{$v_{40}$,...,$v_{51}$\}. The simplices $\{v_1,...,v_{18},v_{40},...,v_{51}\}$ and $\{v_{19},...,v_{39},v_{40},...,v_{51}\}$ are maximal in $K$. This system is thus modeled by a 29-simplex and a 32-simplex joined along a common 11-simplex face, i.e.~this is a \emph{pushout} of those two simplices along the common 11-simplex. 

Another model, and one we will use to perform some calculations in Examples \ref{SSS:EC} and \ref{SSS:ECSS}, is obtained by considering the blue wall and the red wall as single voters. This is reasonable since the assumption is that these states always vote the same way. In this case, the simplicial complex consists of two 12-simplices joined along a common 11-simplex. This is also known as the \emph{suspension} of the 11-simplex.  
\end{example}

\subsection{Power indices for weighted voting games}\label{SS:IndicesSimple}

From now on we assume that a simple game $v$ on voters $N=\{v_1, ..., v_n\}$ may have some unfeasible coalitions and is represented by a simplicial complex $K$ on the vertex set $V=\{v_1, ..., v_n\}$. For a face $\sigma\in K$ containing vertex $v_i$, let $\sigma\vert_{\widehat{v_i}}$ denote the face of $\sigma$ spanned by vertices other than $v_i$.

Then one can redefine the Banzhaf index of a simple game $v$ as (compare with \eqref{E:RelativeBanzhaf})
\begin{equation}\label{E:SimplicialRelativeBanzhaf}
B(v_i)=\frac{\sum_{\sigma\in \text{st}(v_i)} \big(l(\sigma) - l(\sigma\vert_{\widehat{v_i}})  \big)}
{\sum_{v_i\in V}\sum_{\sigma\in \text{st}(v_i)} \big(l(\sigma) - l(\sigma\vert_{\widehat{v_i}})  \big)}
\end{equation}

The Shapley-Shubik index can be recast as (compare with \eqref{E:SS})

\begin{equation}\label{E:SimplicialSS}
SS(v_i)=\frac{\sum_{\sigma\in \text{st}(v)} (\dim \sigma+1)!\cdot(n-\dim \sigma)!\big(l(\sigma) - l(\sigma\vert_{\widehat{v_i}})  \big)}
{\sum_{\substack{\tau \in K_{\max} \\ l(\tau)=1}} (\dim\tau+1)!},
\end{equation}
where $K_{\max}$ is the set of maximal simplices in $K$.

In both cases, the numerator as before counts the number of times a voter is critical/pivotal. The Banzhaf denominator adds up all the times any voter is critical. The denominator of the Shapley-Shubik index scales by the number of permutations of all possible maximal winning coalitions, replacing $n!$ from the usual setup.

In the special case of weighted games, these formulas can be elaborated on further.

\subsection{Banzhaf index for weighted games}\label{SS:Banzhaf}
We first discuss the Banzhaf index for symmetric games since the formulas reduce to a pleasant form in this case.

Recall that a symmetric weighted game is one where all weights are equal, and this is isomorphic to a game with weights 1 and the quota scaled by that common weight. Let $V(q;1,...,1)$ be a symmetric voting game with voters $v_1,...,v_n$ and quota $q$, where $$\frac{n}{2} < q\leq n.$$ 

Let $K$ be the simplicial complex modeling the game with its simplices representing the set of all feasible coalitions in $V$. 
Define the \textit{reduced} $d$-skeleton $\widetilde{K}^d$ of a simplicial complex $K$ as the set $K^d\setminus K^{d-1}$, where $K^d$ is the $d$-skeleton. Note that $\widetilde{K}^d$ is just a subset and not necessarily a subcomplex of $K$.

For a subset $K' \subset K$, let $|K'|$ denote the cardinality of $K'$, i.e.~the number of elements in $K'$.

\begin{prop}
The Banzhaf index $B_s (v_i)$ of the voter $v_i$ in the symmetric weighted game $V(q;1,...,1)$ is given by
$$B_s (v_i)=\frac{|\text{st}(v_i)\cap \widetilde{K}^{q-1}|}{q \cdot |\widetilde{K}^{q-1}|}.$$
\end{prop}
\begin{proof}
A voter is critical precisely when they belong to a feasible coalition of size $q$, and all voters in such a coalition are critical. Represent $V(q;1,...,1)$ as the simplicial complex $K$ whose simplices are all feasible coalitions. Then a winning coalition is a simplex in $K$ containing at least $q$ vertices. But a vertex represents a critical voter only when the simplex has precisely $q$ vertices, i.e.~it is a $(q-1)$-simplex. Thus the reduced $(q-1)$-skeleton $\widetilde{K}^{q-1}$ represents all coalitions where someone is critical and each $(q-1)$-simplex in this skeleton represents $q$ critical voters. Hence the denominator of the Banzhaf index is $q \cdot |\widetilde{K}^{q-1}|.$

The numerator is the number of coalitions in which $v_i$ is critical. The set of all such coalitions is the set of simplices in $\widetilde{K}^{q-1}$ which contain $v_i$; in other words, this is the set of simplices in $K$ which contain $v_i$ as a vertex and which are contained in $\widetilde{K}^{q-1}$. This means that the number of coalitions in which $v_i$ is critical is
$|\text{st}(v_i)\cap \widetilde{K}^{q-1}|.$
\end{proof}

\begin{example}
Let $V$ be a symmetric voting game represented as the simplicial complex from Example \ref{E:complex1}, with quota (a) $q=2$, (b) $q=3$.

In the case (a), the reduced $1$-skeleton $\widetilde{K}^{1}$ is the set $$\{\{v_1, v_2\}, \{v_1, v_3\}, \{v_2, v_3\}, \{v_2, v_4\}, \{v_2, v_5\}, \{v_4, v_5\}, \{v_5, v_6\}\}.$$ Further, we have $c(v_1)=|\text{st}(v_1)\cap \widetilde{K}^{1}|=2$, $c(v_2)=4$, $c(v_3)=2$, $c(v_4)=2$, $c(v_5)=3$, $c(v_6)=1$, $c(v_7)=0$. The denominator for the Banzhaf index is
$$c(v)=\sum_{i=1}^7 c(v_i)=q\cdot |\widetilde{K}^{1}|=14.$$ Therefore, the symmetric Banzhaf indices are $$B_s (v_1)=\frac{2}{14}, B_s (v_2)=\frac{4}{14}, B_s (v_3)=\frac{2}{14}, B_s (v_4)=\frac{2}{14}, B_s (v_5)=\frac{3}{14}, B_s (v_6)=\frac{1}{14}, B_s (v_7)=0.$$

In the case (b), $\widetilde{K}^{2}=\{\{v_1, v_2, v_3\}\}$, $c(v_1)=|\text{st}(v_1)\cap \widetilde{K}^{2}|=1$, $c(v_2)=c(v_3)=1$, $c(v_i)=0$ for $4\leq i \leq 7$, $c(V)=q\cdot |\widetilde{K}^{2}|=3$, and the symmetric Banzhaf indices are $$B_s (v_1)=B_s (v_2)=B_s (v_3)=\frac{1}{3}, B_s (v_i)=0 \text{ for } 4\leq i \leq 7.$$
\end{example}

We now turn to the case of weighted voting games where the weights may not be equal, so now $K$ is a weighted simplicial complex, i.e.~its vertices are labeled by the weights $w_i$. 

\begin{prop}
The Banzhaf index of the voter $v_i$ in a weighted game $V(q;w_1,...,w_n)$ is given by
$$B(v_i)=\frac{\sum_{j=0}^d |\{s\in C_{j}^{v_i} \cap \text{st}(v_i): W_s \geq q\}|}{\sum_{i=1}^n \sum_{j=0}^d |\{s\in C_{j}^{v_i} \cap \text{st}(v_i): W_s \geq q\}|},$$
where $d=\dim K$, $W_s$ is the sum of the weights of the vertices of $s$, and $C_{j}^{v_i}$ is the set of all $j$-simplices in $K$ such that no element in $C_{j}^{v_i} \setminus v_i$ reaches the quota.
\end{prop}

\begin{proof}
The numerator is an expression for $c(v_i)$, the number of times voter $v_i$ is critical, generalizing the symmetric case in a straightforward way by counting up all simplices in $\text{st}(v_i)$ whose weight exceeds the quota but fails to do so without $v_i$. The denominator is obtained by summing up over all vertices as in the usual (relative) Banzhaf index case.
\end{proof}

\begin{subsubsection}{UN Security Council}\label{SSS:UN}
The United Nations Security Council has $15$ members, divided into two subsets we will call PM and NPM, where PM consists of $5$ permanent members (United States, United Kingdom, China, Russia, and France) and NPM consists of $10$ non-permanent members from the rest of the UN. For a resolution to pass in the security council, it needs the support of all $5$ permanent members and at least $4$ non-permanent members. This situation can be modeled as the weighted voting game
$$V(39;7,7,7,7,7,1,1,1,1,1,1,1,1,1,1).$$
The Banzhaf index of any of the countries is easy to calculate and can readily be found in the literature (see, for example, \cite{B:GamePolitics}). For each permanent member, Banzhaf index is approximately 0.1669; and for each non-permanent member, Banzhaf index is approximately 0.0165. Using our approach, we can investigate how the power would shift if certain members decided not to vote together, i.e.~to not be in the same coalition.

The results are summarized in Figure \ref{F:UN-B}. In one case, denoted ``PM vs NPM,'' if a non-permanent member (say NPM1) refused to vote with a permanent member, the Banzhaf power shifts. The PMs have slightly less power ($\approx 0.1582$), and the NPMs have slightly more power ($\approx 0.0232$), except for the NPM1, who now has no power.

In the case ``NPM vs NPM'' of two non-permanent members not voting together, say NPM1 and NPM2, the PMs still have less power ($\approx 0.1610$) than they originally did, and the NPMs still have slightly more power ($\approx 0.0206$) than they originally did, except for the two quarelling NPMs who now have less power ($\approx 0.0150$). 

\begin{figure}[h!]\label{F:UN-B}
\begin{centering}
  \includegraphics[width=16cm]{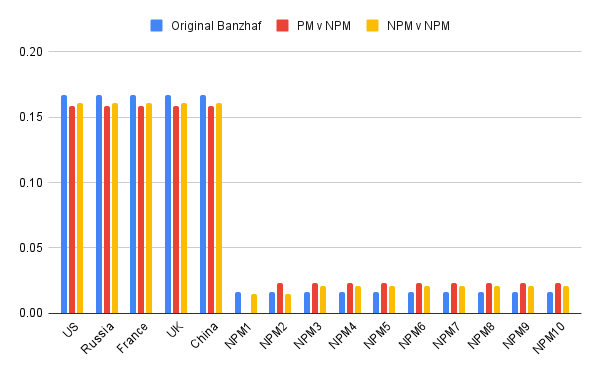}
\caption{Banzhaf indices for the UN Security Council}
  \end{centering}
\end{figure}

\end{subsubsection}

\begin{subsubsection}{Electoral College}\label{SSS:EC}
The calculation of the power distribution in the Electoral College voting system from Example \ref{E:ElectoralC} is also standard, with the Banzhaf power of the individual states ranging from about 11.36\% for California to 0.55\% for the smallest states. We amend the calculation by consolidating the red wall and blue wall states into two variables that hold the collective electoral weights of those states -- 214 for blue and 166 for red. The twelve swing states have the following numbers of electoral votes: Arizona -- 11, Florida -- 30, Georgia  -- 16, Iowa --  6, Maine --  4, Michigan  --  15, Minnesota  --  10, New Hampshire  --  4, North Carolina  --  16, Ohio  --  17, Pennsylvania  --  19, Wisconsin  --  10.

The weighted voting game in question is therefore
$$
V(270;214,166, 30, 19, 17, 16, 16, 15, 11, 10, 10, 6, 4, 4).
$$
Any coalition containing the blue wall and the red wall is unfeasible. As explained in Example \ref{E:ElectoralC}, the simplicial complex representing this system is the suspension of an 11-simplex.

Table 1 contains the calculations (rounded to four decimals) of the Banzhaf index for this system. Note the surprising result that Florida and Pennsylvania have more Banzhaf power than the red wall states combined, even though the red wall has much more weight than either of those states. As we will see, the Shapley-Shubik index does not exhibit this discrepancy, although an argument could be made that the Banzhaf index is the more accurate measure here since the residents of the U.S.~states vote independently and at the same time (not sequentially) in the U.S.~presidential elections. (But this is not the case in the Republican and Democratic primary elections in which each party selects their presidential candidate, so Shapley-Shubik index might be more appropriate there.)\footnote{The discrepancy between weights and power can of course occur for any weighted game, but this may be much more pronounced in the presence of unfeasible coalitions. For example, in the game $V(4; 3, 2, 1)$, if the first and second voter refuse to cooperate, then the second voter has no power. The first and third share the same amount of power since the two of them form the one and only possible winning coalition. Thus the third voter has more power than the second, even if they have less voting weight. In the example at hand, Red States are critical 735 times while Florida is critical 1,258 times and Pennsylvania is critical 800 times.}

\begin{table}\label{T:Electoral}
\begin{center}
\begin{tabular}{| c | c|}
\hline
State & Banzhaf index with unfeasible coalitions\\
\hline
\hline
Blue States & 0.3092\\
\hline
Red States & 0.0685\\
\hline
Florida & 0.1172\\
\hline
Pennsylvania & 0.0745\\
\hline
Ohio & 0.0669\\
\hline
Georgia & 0.0637\\
\hline
North Carolina & 0.0637\\
\hline
Michigan & 0.0594\\
\hline
Arizona & 0.0439\\
\hline
Minnesota & 0.0402\\
\hline
Wisconsin & 0.0402\\
\hline
Iowa & 0.0225\\
\hline
Maine & 0.0149\\ 
\hline
New Hampshire & 0.0149 \\
\hline

\end{tabular}
\end{center}
\caption{Banzhaf indices for the states in the U.S. Electoral College with unfeasible blue wall and red wall coalitions.}
\end{table}

\end{subsubsection}

\begin{subsubsection}{Bosnian-Herzegovinian Parliament}\label{SSS:Bosnia}
Another system we want to analyze is the Parliament of Bosnia and Herzegovina. The political system in that country is a complicated apparatus of nationalist and ethnic parties that have varying allegiances and feuds with each other. After the 2018 general election, there were 14 parties represented within the 42-seat House of Representatives, ranging from 1 to 9 members per party. Counting each party as a single voter with the weight of however many seats they have in the Parliament, we can model this with a weighted voting game $V(22;9,6,5,5,3,3,2,2,2,1,1,1,1,1)$. It is important to point out that this is a simplified model, for the reasons that require an extensive explanation of the Bosnian-Herzegovinian political system. A good reference  is \cite{BGM21}.

There is again a simplicial complex associated to this voting game in which every party is one vertex and the edges between them show party alliances. One can run the standard Banzhaf index or run it with a list of unfeasible pairs of voters, giving a more accurate representation of each party's power.

To make things simpler for running the calculations, we only consider the unfeasible voting pairs of parties with 2 or more votes in the Parliament.\footnote{We use the abbreviations for the names of the political parties. The full names can be found at \url{https://www.izbori.ba/rezultati_izbora?resId=25&langId=1#/2/0/0/0/0/0}.} The unfeasible pairs  are [SDA, SDP], [HDZ, SDP], [SDA, NS], [HDZ, NS], [SDP, SDS], [NS, SDS], [SDP, SNSD], [NS, SNSD], [SNSD, SDS], and [SNSD, PDP]. 

The original Banzhaf indices (all coalitions allowed) and modified Banzhaf indices (with unfeasible pairs) are given in Table 2 and also presented in Figure \ref{F:BosniaBan}. We can see that the power indices for some parties increase after taking into account unfeasible pairs. This is also empirically confirmed: The actual (as in 2022) government is made of SDA (unfeasible Banzhaf index $0.2543$), HDZ (unfeasible index $0.1938$), DF (unfeasible index $0.1014$), and SNSD (unfeasible index $0.0818$). On the other hand, the decision of SDP that it will not be a part of any coalition with nationalist parties (SDA, HDZ, SNSD, SDS) caused the decrease of the index from $0.1152$ to $0$. Same with NS. In other words, SDP and NS are \emph{dummy voters}.

\begin{table}\label{T:BosniaBan}
\begin{center}
\begin{tabular}{| c | c| c |}
\hline
Party & Original Banzhaf Index & Banzhaf index with unfeasibilities\\
\hline
\hline
SDA & 0.2441 & 0.2543\\
\hline
HDZ & 0.1152 & 0.1938\\
\hline
SDP & 0.1152 & 0.0000\\
\hline
DF & 0.0687 & 0.1014\\
\hline
SBB & 0.0452 & 0.0682\\
\hline
NS & 0.0452 & 0.0000\\
\hline
NB & 0.0226 & 0.0338\\
\hline
PDA & 0.0226 & 0.0338\\
\hline
A-SDA & 0.0226 & 0.0338\\
\hline
SNSD & 0.1397 &0.0818\\
\hline
SDS & 0.0687 & 0.0794\\ 
\hline
PDP & 0.0452 &0.0552\\
\hline
DNS & 0.0226 &0.0338\\
\hline
SP & 0.0226 & 0.0338\\
\hline

\end{tabular}
\end{center}
\caption{Original and unfeasible Banzhaf indices in the House of Representatives of Bosnia and Herzegovina.}
\end{table}

\begin{figure}[h!]
\begin{centering}
  \includegraphics[width=16cm]{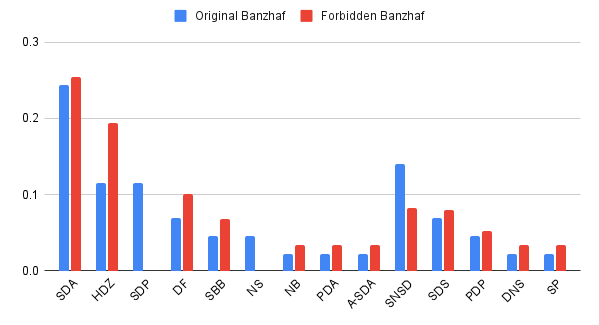}
\caption{Original and unfeasible (forbidden) Banzhaf indices in the House of Representatives of Bosnia and Herzegovina}
\label{F:BosniaBan}
  \end{centering}
\end{figure}

\end{subsubsection}

\subsection{Shapley-Shubik index for weighted games}\label{SS:SS}
We again start with symmetric weighted games. Let $V(q;1,...,1)$ be a symmetric voting game with voters $N=\{v_1, ...,v_n\}$ and possible unfeasible coalitions. Let $K$ be the simplicial complex that represents $V(q;1,...,1)$. 
Denote by $M_j$ the set of all maximal simplices of $K$ with $j+1$ vertices and by $|M_j|$ its cardinality. Let $d=\dim K$. 
\begin{prop}
The Shapley-Shubik index $SS_s (v_i)$ of the voter $v_i$ in the symmetric weighted game $V(q;1,...,1)$ is given by
$$SS_s (v_i)=\frac{\sum_{j=q-1}^d |M_j \cap \text{st}(v_i)|\cdot j!}{\sum_{j=q-1}^d |M_j|\cdot (j+1)!}.$$
\end{prop}
\begin{proof}

The number of all maximal winning coalitions is given by
$$\sum_{j=q-1}^d |M_j|.$$
To compute the number of all winning sequential coalitions, we have to consider every possible order. Since each $m\in M_j$ has $j+1$ vertices, the number of all winning sequential coalitions is obtained by multiplying the above expression by $(j+1)!$. Thus the denominator of the Shapley-Shubik index is given by
$\sum_{j=q-1}^d |M_j|\cdot (j+1)!.$

Recall also that the numerator of the Shapley-Shubik index of the voter $v_i$ is the number of times $v_i$ is pivotal. In the symmetric voting game, $p(v_i)$ is the number of times $v_i$ is the $q^{\text{th}}$ voter to join a coalition. This is given by the number of maximal simplices in $K$ with $q$ or more vertices that contain $v_i$. When $v_i$ is pivotal it is fixed in the $q^{\text{th}}$ sequential position of a coalition, while the order in which the other voters join may vary. Thus the numerator of the Shapley-Shubik index of $v_i$ is given by
$\sum_{j=q-1}^d |M_j \cap \text{st}(v_i)|\cdot j!.$
%
%
\end{proof}

\begin{example}
Let $V$ be a symmetric voting game represented as the simplicial complex from Example \ref{E:complex1}, with quota (a) $q=2$, (b) $q=3$. Note that the dimension of the complex is $d=2$.

In the case (a), $$M_1=\{\{v_1, v_2\}, \{v_1, v_3\}, \{v_2, v_3\}, \{v_2, v_4\}, \{v_2, v_5\}, \{v_4, v_5\}, \{v_5, v_6\}\},$$ and $$M_2=\{\{v_1,v_2,v_3\}\}.$$
The Shapley-Shubik index for the voter $v_1$ is
$$SS_s (v_1)=\frac{\sum_{j=1}^2 |M_j \cap \text{st}(v_1)|\cdot j!}{\sum_{j=1}^2 |M_j|\cdot (j+1)!}=\frac{2\cdot 1! + 1\cdot 2!}{7\cdot 2!+1\cdot 3!}=\frac{4}{20},$$
and, similarly,
$$SS_s (v_2)=\frac{6}{20}, SS_s (v_3)=\frac{4}{20}, SS_s (v_4)=\frac{2}{20}, SS_s (v_5)=\frac{3}{20}, SS_s (v_6)=\frac{1}{20}, SS_s (v_7)=0.$$

In the case (b), the only relevant are $2$-simplices, i.e. the set $M_2=\{\{v_1,v_2,v_3\}\}$. The Shapley-Shubik indices are
\begin{align*}
SS_s (v_1)=\frac{|M_2 \cap \text{st}(v_1)|\cdot 2!}{|M_2|\cdot 3!}=\frac{1\cdot 2!}{1\cdot 3!} & =\frac{1}{3};\\
SS_s (v_2)=SS_s (v_3) & =\frac{1}{3}; \\
 SS_s (v_i) & =0,\  \text{ for } 4\leq i \leq 7.
\end{align*}
\end{example}

Let $V(q;w_1,...,w_n)$ now be a weighted voting game represented by the $d$-dimensional simplicial complex $K$ whose vertices are labeled by voter weights and simplices are all possible coalitions. Let $M_j$ again  denote the set of all maximal simplices in $K$ with $j+1$ vertices. For a simplex $s$, $W_s$ denotes the sum of the weights of the vertices of $s$. 
\begin{thm}
Let $\{f_1,...,f_m\}$ be the set of all maximal simplices that contain $v_i$. For a maximal simplex $f_k$, let $P_{j}^{f_k, v_i}$ be the set of all $j$-dimensional faces of $f_k$ such that for an arbitrary $\varphi \in P_{j}^{f_k, v_i}$ no $(j-1)$-dimensional face of $\varphi$ in $\text{st}(v_i)$ reaches the quota, and let $n_{f_k}$ denote the number of vertices in $f_k$. Then the Shapley-Shubik index of the voter $v_i \in V$ in a weighted game $V(q;w_1,...,w_n)$ is given by
$$SS(v_i)=\frac{\sum_{k=1}^{m} \sum_{j=0}^{n_{f_k}-1} |\{s\in P_{j}^{f_k, v_i} \cap \text{st}(v_i): W_s \geq q\}|\cdot j! \cdot (n_{f_k}-(j+1))!}{\sum_{j=0}^d |\{s\in M_j: W_s \geq q\}|\cdot (j+1)!}.$$
\end{thm}

\begin{proof}
We can represent the maximal winning coalitions as the set of all maximal simplices in $K$ that contain vertices whose weights sum up to $q$ or greater. So the number of maximal winning coalitions is given by
$$\sum_{j=0}^d |\{s\in M_j: W_s \geq q\}|.$$

The denominator for the Shapley-Shubik index, being the number of all maximal sequential coalitions, is then given by
$$\sum_{j=0}^d |\{s\in M_j: W_s \geq q\}|\cdot (j+1)!.$$

The numerator for the Shapley-Shubik index of a voter $v_i$ in $V$ is the number of times $v_i$ is pivotal. Since the Shapley-Shubik index only considers maximal winning coalitions, in order to compute the index we can break down our simplicial complex $K$ into its maximal simplices and work through each maximal simplex separately. 

Note that for each maximal simplex $f=\{v_{j_1},...,v_{j_{n_f}}\}$, where $n_f$ is the number of vertices in $f$, a voter $v_i \in f$ is pivotal if the following is satisfied:
\begin{enumerate}
\item[(i)] $i=j_k$ for a $k\in \{1,...,n_f\};$
\item[(ii)]  the total weight of the face $\{v_{j_1},...,v_{j_k-1}\}$ does not reach the quota, i.e.
$$w_{j_1}+\cdots+w_{j_k-1}<q;$$
\item[(iii)]  the total weight of the face $\{v_{j_1},...,v_{j_k}\}$ reaches the quota, i.e.
$$w_{j_1}+\cdots+w_{j_k}\geq q.$$
\end{enumerate}

Hence we must only count the faces of $f$ containing $v_i$ that have weight greater or equal to $q$, but would have weight less than $q$ if we removed $v_i$. So the number of times $v_i$ is pivotal in a maximal sequential coalition $f$ is given by
$$\sum_{j=0}^{n_f-1} |\{s\in P_{j}^{f, v_i} \cap \text{st}(v_i): W_s \geq q\}|\cdot j! \cdot (n_f-(j+1))!,$$
where $P_{j}^{f, v_i}$ is the set of all $j$-dimensional faces of the maximal simplex $f$ such that for an arbitrary $\varphi \in P_{j}^{f, v_i}$, no $(j-1)$-dimensional face of $\varphi$ in $\text{st}(v_i)$ reaches the quota.
We multiply by $j!$ and $(n_f-(j+1))!$ to count the number of orders in which voters join a coalition before and after $v_i$.

To get the total number of times $v_i$ is pivotal in $V$, we evaluate the above expression for each maximal simplex that contains $v_i$ and take the sum of the resulting value. Thus, if $\{f_1,...,f_m\}$ is the set of all maximal simplices that contain $v_i$, the numerator for the Shapley-Shubik index of $v_i$ equals
$$\sum_{k=1}^{m} \sum_{j=0}^{n_{f_k}-1} |\{s\in P_{j}^{f_k, v_i} \cap \text{st}(v_i): W_s \geq q\}|\cdot j! \cdot (n_{f_k}-(j+1))!$$
as desired.
%
%
\end{proof}


\begin{subsubsection}{UN Security Council}
Consider the United Nations Security Council one more time (see Example \ref{SSS:UN}). With the Shapley-Shubik index, the permanent members have considerably more power than the non-permanent members; unrestricted, each PM has $19.6\%$  of the voting power while each NPM has $0.19\%$ of the voting power. As we consider two cases, one where a PM and NPM1 are quarelling, and the other where NPM1 and NPM2 are quarelling, the power shifts, but only by about a fifth of a percent. See Figure 4.

\begin{figure}\label{F:UN-SS}
\begin{centering}
  \includegraphics[width=16cm]{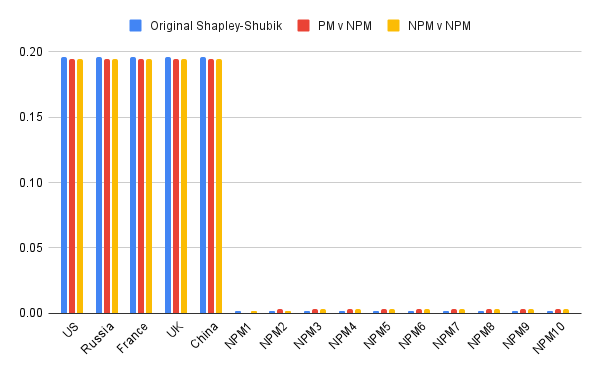}
\caption{Shapley-Shubik indices for the UN Security Council}
  \end{centering}
\end{figure}

\end{subsubsection}

\begin{subsubsection}{Electoral College}\label{SSS:ECSS}
With no restrictions on coalitions, the Shapley-Shubik indices of the voters in the Electoral College are also well-known, ranging from 11.03\% for California to 0.54\% for states with three electoral votes.  This calculation is especially taxing because of the large number of sequential coalitions. By consolidating the red and blue states, this problem becomes much more tractable while reflecting a more realistic situation. The resulting Shapley-Shubik indices can be found in Table 3.

\begin{table}\label{T:ECSS}
\begin{center}
\begin{tabular}{| c | c|}
\hline
State & Shapley-Shubik index  with unfeasible coalitions\\
\hline
\hline
Blue States & 0.3173\\
\hline
Red States & 0.1799\\
\hline
Florida & 0.1074\\
\hline
Pennsylvania & 0.0595\\
\hline
Ohio & 0.0524\\
\hline
Georgia & 0.0502\\
\hline
North Carolina & 0.0502\\
\hline
Michigan & 0.0468\\
\hline
Arizona & 0.0346\\
\hline
Minnesota & 0.0319\\
\hline
Wisconsin & 0.0319\\
\hline
Iowa & 0.0168\\
\hline
Maine & 0.0105\\ 
\hline
New Hampshire & 0.0105 \\
\hline

\end{tabular}
\end{center}
\caption{Shapley-Shubik indices for the states in the U.S. Electoral College with unfeasible blue wall and red wall coalitions.}
\end{table}

\end{subsubsection}

\begin{subsubsection}{Bosnian-Herzegovinian Parliament}
Consider the voting game from Example \ref{SSS:Bosnia}. Because there are so many quarrelling parties, the number of winning permutations is considerably less than the original index, so the voting power shifts considerably as well. Original Shapley-Shubik indices (with no unfeasible pairs) and unfeasible Shapley-Shubik indices (with unfeasible pairs as defined  in Example \ref{SSS:Bosnia}) are given in Table 4 and also presented in Figure 5. The result is similar to that of the Banzhaf indices.

\begin{table}\label{T:BosniaSS}
\begin{center}
\begin{tabular}{| c | c| c |}
\hline
Party & Original Shapley-Shubik index & Unfeasible (forbidden) Shapley-Shubik index\\
\hline
\hline
SDA & 0.2439 & 0.3940\\
\hline
HDZ & 0.1178 & 0.1816\\
\hline
SDP & 0.1178 & 0.0000\\
\hline
DF & 0.0675 & 0.0905\\
\hline
SBB & 0.0441 & 0.0560\\
\hline
NS & 0.0441 & 0.0000\\
\hline
NB & 0.0217 & 0.0255\\
\hline
PDA & 0.0217 & 0.0255\\
\hline
A-SDA & 0.0217 & 0.0255\\
\hline
SNSD & 0.1450 &0.0155\\
\hline
SDS & 0.0675 & 0.0832\\ 
\hline
PDP & 0.0441 &0.0516\\
\hline
DNS & 0.0217 &0.0255\\
\hline
SP & 0.0217 & 0.0255\\
\hline
\end{tabular}
\end{center}
\caption{Original and unfeasible Shapley-Shubik indices in the House of Representatives of Bosnia and Herzegovina}
\end{table}

\begin{figure}\label{F:BosniaSS}
\begin{centering}
  \includegraphics[width=16cm]{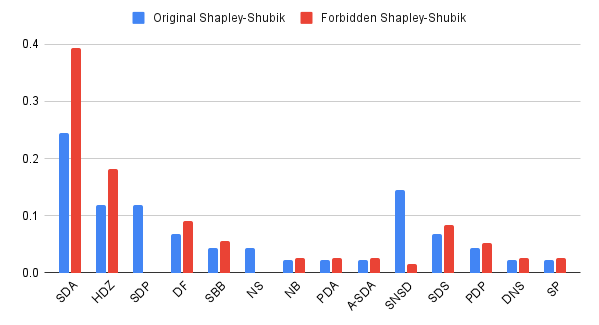}
\caption{Original and unfeasible (forbidden) Shapley-Shubik indices in the House of Representatives of Bosnia and Herzegovina}
  \end{centering}
\end{figure}

\end{subsubsection}

\section{Further work}\label{S:FurtherResearch}

One of the most obvious directions of potential future investigation is the interaction of the topology of simplicial complexes with power indices. Notions such as the barycentric subdivision, dual complex, poset of a complex, wedge sum, cone, and suspension could all provide refined insight into roles and power shares of certain voters based on their location in the simplex. Simplicial maps could supply a category-theoretic and functorial view on simple games. Some work in this direction was done in \cite{MT:MathGames}. This can be combined with algebraic topology ideas of homology and contiguity for the analysis of the dynamics of power distribution as voters relinquish their vote.  Many such correspondences between the topology of simplicial complexes and political structures were explored in \cite{MV:Politics}, and the same could be done here.

One ungratifying feature of simplicial complexes is the mandate that, if a coalition is feasible, then all its subcoalitions must also be feasible (since all faces of each simplex in a simplicial complex must also be in it). But this does not necessarily reflect reality, as some coalitions might dissolve if certain voters are not in them. This possibility would be captured with a \emph{hypergraph} model; these are generalizations of simplicial complex in precisely the necessary -- if a simplex is in a complex, there is no requirement that its faces are as well. Hypergraphs are well-studied from the graph-theoretic point of view, and one could bring this into simple games. In particular, it would be interesting to study and reconcile the interaction of the \emph{clustering coefficient}, a notion of ``importance'' of a vertex, with its Banzhaf and Shapley-Shubik power.

As discussed in \cite{KHJ18}, the traditional method for assigning weights to a voter may not accurately represent the power of a member in a voting system because it does not consider the diplomatic influence of the member. The authors use an association matrix to quantify the ``persuasive power'' of a voter and use the entries of that matrix to improve the Banzhaf index. We have developed formulas to calculate power indices in a voting game such that certain coalitions may be unfeasible, but reality suggests that the probability of coalition formation may not always be binary. One reason for this is that a group of voters may agree on some issues, but not on others. In future work, we want to pursue this line of reasoning, but rather than considering the ``persuasive power'' of a voter, we plan to consider the probability that a voter will join a coalition. One should be able to integrate this probability into the formulas developed in Section 3 without too much trouble. Here is a rough idea.

Let $V(q;w_1,...,w_n)$ be a weighted voting game with voters $v_1,v_2,...,v_n$. Let the $n \times n$ matrix $A$ have entries $a_{ij}$ such that $a_{ij}\in [0,1]$ is the probability that $v_i$ will form a coalition with $v_j$, for $i,j \in \{1,...n\}$. Let $C=\{v_{j_1},...,v_{j_k}\}$ be a coalition in $V$ such that $v_i \notin C$. Let the $n\times 1$ matrix $\Omega$ have entries $\omega_{i,C}$ such that $\omega_{i,C}$ is the probability that $v_i$ will join coalition $C$. We define $\omega_{i,C}$ to be the average of all $a_{ij_m}$ for $j_m\in \{1,...k\}$, that is,
$$\omega_{i,C} = \frac{\sum_{m=1}^k a_{ij_m} }{k},$$
where $k=|C|$. Using this setup, one should be able to develop a formula for the improved Banzhaf index of a voter $v_i$ in a weighted voting game in terms of the probabilities $\omega_{i,C}$ for suitably defined $C$'s via simplices. In a similar manner, one can improve the Shapley-Shubik index formulas.




Instead of using probabilities, vertices representing agents might have ``agreement matrices'' attached to them. These would indicate which issues the agents agree or disagree on. The entries might even reflect the levels of agreement, i.e.~they might take values in the real numbers rather than in the set $\{0,1\}$. This setup is reminiscent of a \emph{sheaf}. It is plausible that sheaf-theoretic techniques (such as the Laplacian of sheaf cohomology) could be used to define entirely new power indices that would capture much more of the political dynamics between voters or political parties.

Another interesting direction would be to extend our dictionary to \emph{multichoice}, or \emph{$j$-cooperative games} \cite{F:BanzhafMultichoice, FZ:WeightedMultichoice}. In this setup, a voter has multiple action options, rather than just a binary choice.

There are other indices besides the Banzhaf and Shapley-Shubik, and they measure voting power slightly differently. One of interest is the Deegan-Packel index \cite{DP78}, which considers minimal coalitions, those just big enough to pass a quota. There is merit to this index because often coalitions do not need to expend the extra energy to get more votes than necessary. We expect that our simplicial complex model would extend to this index fairly readily. 

Finally, one could continue the analysis started in \refT{CompleteWeightedGames} and \refC{CompleteHomologyWeighted}. Characterizing weighted games is a big topic in game theory and having a topological approach to this question would be extremely useful. Zooming out, there is likely a lot that topology can say about simple games in general, starting with the basic dictionary set forth in \refP{Games<->SCs}. For example, Alexander duality is well-studied and well-understood in topology, and the question is whether some of that knowledge can translate into insight about simple games. Same for some standard constructions on simplicial complexes mentioned at the top of this section, like join and suspension.

\bibliographystyle{alpha}



\end{document}